\documentclass[11pt,a4paper]{article}    % onecolumn (standardformat)

\usepackage[margin=1in]{geometry} %In Folder
\usepackage{amsmath,amssymb,amsthm}
\usepackage{color}
\usepackage{relsize}
\usepackage[perpage,symbol*]{footmisc}
\usepackage{authblk}

\newtheorem{theorem}{Theorem}[section]
\newtheorem{proposition}[theorem]{Proposition}
\newtheorem{lemma}[theorem]{Lemma}
\newtheorem{corollary}[theorem]{Corollary}

\newtheorem{remark}{Remark}

\newcommand{\be}{\begin{equation}}
\newcommand{\ee}{\end{equation}}
\newcommand{\bea}{\begin{eqnarray}}
\newcommand{\eea}{\end{eqnarray}}
\newcommand{\e}{\emph e}

\numberwithin{equation}{section}

\begin{document}

\title{The Smallest Eigenvalue Distribution of the \\Jacobi Unitary Ensembles}

\author{Shulin Lyu\thanks{School of Mathematics (Zhuhai), Sun Yat-sen University, Zhuhai 519082, China; e-mail: lvshulin1989@163.com} ~and Yang Chen\thanks{Department of Mathematics, Faculty of Science and Technology, University of Macau, Macau, China; e-mail: yangbrookchen@yahoo.co.uk}}

\date{}

\maketitle
\begin{abstract}
In the hard edge scaling limit of the Jacobi unitary ensemble generated by the weight $x^{\alpha}(1-x)^{\beta},~x\in[0,1],~\alpha,\beta>0$,
the probability that all eigenvalues of Hermitian matrices from this ensemble lie in the interval $[t,1]$ is given by the Fredholm determinant of the Bessel kernel. We derive the constant in the asymptotics of this Bessel-kernel determinant. A specialization of the results gives the constant in the asymptotics of the probability that the interval $(-a,a),a>0,$ is free of eigenvalues in the Jacobi unitary ensemble with the symmetric weight $(1-x^2)^{\beta}, x\in[-1,1]$.
\end{abstract}

$\mathbf{Keywords}$: Jacobi unitary ensemble; Smallest eigenvalue distribution; Asymptotic

expansions; Bessel kernel; Fredholm determinant.

$\mathbf{Mathematics\:\: Subject\:\: Classification\:\: 2020}$: 15B52; 34E05; 41A60; 42C05.

\section{Introduction}
In this paper, we consider the Jacobi unitary ensemble (JUE for short) of $n\times n$ Hermitian matrices whose eigenvalues have the following probability density function
\begin{align}\label{Jp}
p(x_1,x_2,\ldots,x_n)=\frac{1}{C_n}\prod_{1\leq i<j\leq n}(x_i-x_j)^2\prod_{k=1}^{n}w(x_k,\alpha,\beta),
\end{align}
where $w(x,\alpha,\beta)$ is the modified Jacobi weight reading
\[w(x,\alpha,\beta)=x^{\alpha}(1-x)^{\beta},\qquad x\in[0,1],\quad\alpha,\beta>0.\]
Here $C_n$ is the normalization constant which has a closed-form expression
\begin{align*}
C_n:=&\int_{[0,1]^n} \prod_{1\leq i<j\leq n}(x_i-x_j)^2\prod_{k=1}^{n}w(x_k,\alpha,\beta)dx_k\\
=&\prod_{j=0}^{n-1}\frac{\Gamma(2+j)\Gamma(\alpha+1+j)\Gamma(\beta+1+j)}{\Gamma(2)\Gamma(\alpha+\beta+1+n+j)}.
\end{align*}
See \cite[formula (17.1.3)]{Mehta}.

The probability that all eigenvalues of this unitary ensemble lie in the interval $[t,1]$, or equivalently the smallest eigenvalue is not less than $t$, is given by
\begin{align}\label{PD}
\mathbb{P}(t,\alpha,\beta,n)=&\frac{D_n(t,\alpha,\beta)}{D_n(0,\alpha,\beta)},
\end{align}
where $D_n(t,\alpha,\beta)$ is the following multiple integral
\begin{align*}
D_n(t,\alpha,\beta):=\frac{1}{n!}\int_{[t,1]^n}\prod\limits_{1\leq i<j\leq n}(x_i-x_j)^2\prod\limits_{k=1}^n w(x_k,\alpha,\beta)dx_k.
\end{align*}
Obviously $D_n(0,\alpha,\beta)=C_n/n!$. By changing variables $x_{\ell}=(1-y_{\ell})/2,\ell=1,2,\cdots,n$ in $D_n(t,\alpha,\beta)$, we have
\begin{align*}
D_n(t,\alpha,\beta)=2^{-n(n+\alpha+\beta)}\frac{1}{n!}\int_{[-1,1-2t]^n}\prod\limits_{1\leq i<j\leq n}(y_i-y_j)^2\prod\limits_{k=1}^n (1-y_k)^{\alpha}(1+y_k)^{\beta}dy_k.
\end{align*}
Hence, from \eqref{PD} we find that $\mathbb{P}(t,\alpha,\beta,n)$ is connected with the largest eigenvalue distribution of JUE with the weight $(1-x)^{\alpha}(1+x)^{\beta}$ by
\begin{align}\label{Pconnection}
\mathbb{P}(t,\alpha,\beta,n)=\widehat{\mathbb{P}}(1-2t,\alpha,\beta,n),
\end{align}
where $\widehat{\mathbb{P}}(1-2t,\alpha,\beta,n)$ is the probability that all eigenvalues of $n\times n$ Hermitian matrices from JUE with the weight $(1-x)^{\alpha}(1+x)^{\beta}$ lie in the interval $[-1,1-2t]$.

By representing the associated orthogonal polynomials as a solution of a Riemann-Hilbert problem which was analyzed using Deift-Zhou nonlinear steepest descent method \cite{DeiftZhou} (we call it RH method for short below), Kuijlaars and Vanlessen \cite[Corollary 1.2]{KV2002} proved that
\begin{align}\label{PJbessel}
\lim_{n\rightarrow\infty}\widehat{\mathbb{P}}\left(1-\frac{s}{2n^2},\alpha,\beta,n\right)=\det\left(I-K_{\mathrm{Bessel}}\right),
\end{align}
where $\det\left(I-K_{\mathrm{Bessel}}\right)$ is the Fredholm determinant and $K_{\mathrm{Bessel}}$ is the integral operator with the Bessel kernel
\begin{align}\label{Besal}
K_{\mathrm{Bessel}}(x,y):=\frac{J_{\alpha}(\sqrt{x})\sqrt{y}J_{\alpha}'(\sqrt{y})-\sqrt{x}J_{\alpha}'(\sqrt{x})J_{\alpha}(\sqrt{y})}{2(x-y)},
\end{align}
acting on $L^2(0,s)$. Therefore, it follows from \eqref{Pconnection} and \eqref{PJbessel} that
\begin{align}\label{limPBes}
\lim_{n\rightarrow\infty}\mathbb{P}\left(\frac{s}{4n^2},\alpha,\beta,n\right)=\det\left(I-K_{\mathrm{Bessel}}\right),
\end{align}
This Bessel-kernel determinant also describes the smallest eigenvalue distribution of Laguerre unitary ensembles \cite{Forrester1993, TW161}.

It is well known that the gap probability on $(a,b)$, i.e. the probability that the interval $(a,b)$ has no eigenvalues of unitary ensembles is given by the Fredholm determinant
\[\det\left(I-K_n\mathlarger {\chi}_{(a,b)}\right),\]
where $\mathlarger{\mathlarger {\chi}}_{(a,b)}(\cdot)$ is the characteristic function of the interval $(a,b)$ and the integral operator $K_{n}\mathlarger{\mathlarger {\chi}}_{(a,b)}$  has kernel $K_n(x,y)\mathlarger{\mathlarger {\chi}}_{(a,b)}(y)$.
The kernel $K_n(x,y)$ tends to the sine kernel by scaling the bulk of the spectrum of Gaussian unitary ensembles \cite{TW159, TW161}. The constant in the asymptotics of the corresponding Fredholm determinant was conjectured by Dyson \cite{Dyson} to be $(\log 2)/12+3\zeta'(-1),$ with $\zeta(\cdot)$ denoting the Riemann zeta function, and was proved in \cite{DIKZ2007, Krasovsky2004} via RH method and in \cite{Ehrhardtsine} by using operators. At the soft edge of spectrum of Gaussian and Laguerre unitary ensembles, the kernel $K_n(x,y)$ is the Airy kernel \cite{Forrester1993}. The constant term in the asymptotics of the Fredholm determinant was conjectured by Tracy and Widom \cite{TW161} to be $(\log2)/24+\zeta'(-1)$, and was proved in \cite{DIK2008} by means of RH method, in \cite{BBD} by using an integral expression of the Tracy-Widom distribution, and in \cite{LyuMinChen2020} via the same derivation framework as this paper.

In addition to being applied to derive the constant term in the asymptotics of the gap probability of unitary ensembles, RH method is a powerful tool to study many other problems in large dimensional unitary ensembles, for example, the partition function \cite{ACM,DaiXuZhang2019}, the gap probability distribution \cite{DaiXuZhang2020,XuZhao2019}, the correlation kernel \cite{ChenChenFan}, orthogonal polynomials \cite{ChenChenFan2019}. For finite $n$ analysis, the ladder operators adapted to monic orthogonal polynomials are usually used. See for instance \cite {BasorChenZhang2012,Lyu2019,MinChen2019,MinChen2020}.

Our main purpose in this paper is to derive the constant term in the asymptotic expansion of $\mathbb{P}\left(\frac{s}{4n^2},\alpha,\beta,n\right)$, i.e. the Fredholm determinant of the Bessel kernel (refer to \eqref{limPBes}), in large $s$ with $n$ sufficiently large, which was conjectured by Tracy and Widom \cite{TW161} to be
\[\frac{G(\alpha+1)}{(2\pi)^{\alpha/2}}.\]
It was proved rigorously by Ehrhardt \cite{Ehrhardtbessel} for $|{\rm Re}~\alpha|<1$ via operator approach and by Deift, Krasovsky and Vasilevska \cite{DKV2010} for ${\rm Re}~a>-1$ through RH method. Our derivation is elementary, mainly based on the properties of orthogonal polynomials.

By means of the ladder operator approach, Chen and Zhang \cite{ChenZhang2010} investigated the Hankel determinant generated by the Jacobi weight with a jump, i.e. $x^{\alpha}(1-x)^{\beta}(A+B\theta(x-t)),x\in[0,1],$
with $A\geq0,~ A+B\geq0$, and $\theta(\cdot)$ denoting the Heaviside step function. We observe that the case where $A=0$ and $B=1$ corresponds to $D_n(t,\alpha,\beta)$ of our interest.
Since $D_n(0,\alpha,\beta)$ has an explicit expression, to study $\mathbb{P}(t,\alpha,\beta,n)$ which is equal to the quotient of $D_n(t,\alpha,\beta)$ and $D_n(0,\alpha,\beta)$, we make use of the results in \cite{ChenZhang2010} to deal with $D_n(t,\alpha,\beta)$.

Following the approach in \cite{LyuMinChen2020} which is motivated by the analysis in \cite{DIK2008} and \cite{DIKZ2007} , we study the asymptotic behavior of $\mathbb{P}(t,\alpha,\beta,n)$ from two aspects:
\begin{itemize}
\item [(i)] Using the definition of $D_n(t,\alpha,\beta)$, i.e. the multiple integral, we derive its asymptotic formula for $t\rightarrow1$ with $n$ fixed. Combining the result with the expression of $D_n(0,\alpha,\beta)$, and approximating the involved special function for large $n$, we deduce the asymptotic expression for $\log \mathbb{P}(t,\alpha,\beta,n)$ with $n$ large and $t\rightarrow1$. This step is done in the next section.
\item [(ii)] Based on the finite $n$ results in \cite{ChenZhang2010}, we establish an asymptotic expansion for $(d/dt)\log D_n$, i.e. $(d/dt)\log \mathbb{P}_n(t,\alpha,\beta)$ in large $n$ with $t$ fixed. Section 3 is devoted to the analysis of this part.
\end{itemize}
Integrating the result in (ii) from $t$ to $t_0$ with $t_0$ near 1, and applying the result in (i) to $\log \mathbb{P}(t_0,\alpha,\beta,n)$, by sending $t_0$ to 1, we produce in section 4 an approximate expression for $\log\mathbb{P}(t,\alpha,\beta,n)$ with $n$ large and $t$ arbitrary.
By taking $t=s/(4n^2)$ and letting $n$ tend to $\infty$, we obtain the asymptotic expansion in large $s$ for the hard edge scaling limit of $\log\mathbb{P}(t,\alpha,\beta,n)$ which, as mentioned before, is also equal to the log of the Fredholm determinant of the Bessel kernel.

\begin{theorem}\label{main1}
For $\alpha>-1$ and as $n\rightarrow\infty$, the probability that all eigenvalues of JUE with the weight $x^{\alpha}(1-x)^{\beta}$ lie in the interval $[s/(4n^2),1]$, which is equal to the Fredholm determinant of the Bessel kernel, has the following asymptotic expression
\begin{equation}\label{Pfullexpan}
\begin{aligned}
\lim_{n\rightarrow\infty}\log \mathbb{P}&\left(\frac{s}{4n^2},\alpha,\beta,n\right)=\log\det(I-K_{\mathrm{Bessel}})\\
=&-\frac{s}{4}+\alpha \sqrt{s}-\frac{\alpha^2}{4}\log s+\log \frac{G(\alpha+1)}{(2\pi)^{\alpha/2}}+\frac{\alpha }{8}s^{-1/2}\\
&+\frac{\alpha
   ^2}{16}s^{-1}+\left(\frac{\alpha^3}{24}+\frac{3\alpha}{128}\right)s^{-3/2}+\left(\frac{\alpha^4}{32}+\frac{9\alpha^2}{128}\right)s^{-2}\\
&
+\left(\frac{\alpha^5}{40}+\frac{9\alpha^3}{64}+\frac{45\alpha}{1024}\right)s^{-5/2}+O\left(s^{-3}\right),\qquad s\rightarrow\infty.
\end{aligned}
\end{equation}
\end{theorem}

As was demonstrated in \cite{LyuChenFan}, the gap probability on the interval $(-a,a),a>0,$ of JUE with the weight $(1-x^2)^{\beta}$ is intimately related to the smallest eigenvalue distribution of JUE of our interest with $\alpha=\pm1/2$. From the above theorem, we obtain below the asymptotics.
\begin{corollary}\label{main2}
The probability that the interval $\left(-b/n,b/n\right)$, with $b>0$ and $n\rightarrow\infty$, contains no eigenvalues of JUE with the weight $(1-x^2)^{\beta}$ has the following asymptotic expansion
\begin{equation}\label{PJs}
\begin{aligned}
\lim_{n\rightarrow\infty}\log\widetilde{\mathbb{P}}\left(\frac{b}{n},\beta,n\right)
=&-\frac{b^2}{2}-\frac{\log b }{4}+\frac{\log 2}{12}+3\zeta'(-1)+\frac{1}{32\,b^{2}}\\
&+\frac{5}{128\,b^{4}}+O\left(b^{-6}\right),\qquad b\rightarrow\infty.
\end{aligned}
\end{equation}
\end{corollary}

\section{Asymptotics of $\mathbb{P}(t,\alpha,\beta,n)$ with $n$ Large and $t$ near 1}
Recall that the smallest eigenvalue distribution function on $[t,1$] for JUE can be written as the quotient of two Hankel determinants, that is
\begin{align*}
\mathbb{P}(t,\alpha,\beta,n)=&\frac{D_n(t,\alpha,\beta)}{D_n(0,\alpha,\beta)},
\end{align*}
where $D_n(t,\alpha,\beta)$ is a multiple integral which, according to Heine's formula, can be represented as the determinant of a Hankel matrix \cite{Szego2003}, namely
\begin{align*}
D_n(t,\alpha,\beta):=&\frac{1}{n!}\int_{[t,1]^n}\prod\limits_{1\leq i<j\leq n}(x_i-x_j)^2\prod\limits_{k=1}^n x_k^{\alpha}(1-x_k)^{\beta}dx_k\\
=&\det\left(\int_t^1 x^{i+j}x^{\alpha}(1-x)^{\beta}dx\right)_{i,j=0}^{n-1}.
\end{align*}

According to formula (17.1.3) in \cite{Mehta}, we know that $D_n(0,\alpha,\beta)$ has the following explicit representation
\begin{align}
D_n(0,\alpha,\beta)=&\frac{1}{n!}\prod_{j=0}^{n-1}\frac{\Gamma(2+j)\Gamma(\alpha+1+j)\Gamma(\beta+1+j)}{\Gamma(2)\Gamma(\alpha+\beta+1+n+j)}\nonumber\\
=&\frac{G(n+1)G(n+\alpha+1)G(n+\beta+1)G(n+\alpha+\beta+1)}{G(\alpha+1)G(\beta+1)G(2n+\alpha+\beta+1)},\label{Dn0G}
\end{align}
where $G(\cdot)$ is the Barnes-G function defined by
\begin{align*}
G(z+1)=\Gamma(z)G(z),\qquad\qquad G(1):=1.
\end{align*}
See \cite{Voros} for more properties of this function.

Now we look at $D_n(t,\alpha,\beta)$.
By changing variables $x_\ell=t+(1-t)y_{\ell}, ~\ell=1,2,\cdots,n$, we find
\begin{align*}
D_n(t,\alpha,\beta)=(1-t)^{n(n+\beta)}\cdot t^{n\alpha}\cdot\frac{1}{n!}\int_{[0,1]^n}\prod\limits_{1\leq i<j\leq n}(y_i-y_j)^2\prod\limits_{k=1}^n \left[1+\left(\frac{1}{t}-1\right)y_k\right]^{\alpha}(1-y_k)^{\beta}dy_k.
\end{align*}
Note that, as $t\rightarrow1$, we have
\begin{align*}
\left[1+\left(\frac{1}{t}-1\right)y_k\right]^{\alpha}=1+\alpha\,y_k(1-t)+\frac{\alpha}{2}y_k((\alpha-1)y_k +2)(1-t)^2+O_{\alpha}((1-t)^3),
\end{align*}
where $O_{\alpha}((1-t)^3)$ depends on $\alpha$ and is a third order infinitesimal with respect to $1-t$. Hence, $D_n(t,\alpha,\beta)$ can be approximated by
\begin{align*}
D_n(t,\alpha,\beta)=(1-t)^{n(n+\beta)} t^{n\alpha}D_n(0,0,\beta)\left(1+O_{n,\alpha,\beta}(1-t)\right),
\end{align*}
where $O_{n,\alpha,\beta}(1-t)$ is a first order infinitesimal with respect to $1-t$ and it depends on $n,\alpha$ and $\beta$. Then it follows from the fact $\mathbb{P}(t,\alpha,\beta,n)=D_n(t,\alpha,\beta)/D_n(0,\alpha,\beta)$ that
\begin{align}
\log\mathbb{P}(t,\alpha,\beta,n)=&\log D_n(t,\alpha,\beta)-\log D_n(0,\alpha,\beta)\nonumber\\
=&n(n+\beta)\log(1-t)+n\alpha\log t+\log\frac{D_n(0,0.\beta)}{D_n(0,\alpha,\beta)}+\log((1+O_{n,\alpha,\beta}(1-t))).\label{logp1}
\end{align}
\begin{remark} Formula \eqref{logp1} is an interesting formula where there appears $D_n(0,0,\beta)$. The quantity $D_n(0,0,\beta)/D_n(0,\alpha,\beta)$ will give rise to constant terms depending on $\alpha,\beta$ and $n$.
\end{remark}

On setting $\alpha=0$ in \eqref{Dn0G}, we get
\begin{align*}
D_n(0,0,\beta)=\frac{G^2(n+1)G^2(n+\beta+1)}{G(\beta+1)G(2n+\beta+1)}.
\end{align*}
Dividing it by \eqref{Dn0G} yields
\begin{align*}
\frac{D_n(0,0,\beta)}{D_n(0,\alpha,\beta)}=G(\alpha+1)\cdot\frac{G(n+1)G(n+\beta+1)G(2n+\alpha+\beta+1)}{G(n+\alpha+1)G(2n+\beta+1)G(n+\alpha+\beta+1)}.
\end{align*}
Then, from \eqref{logp1}, it follows that
\begin{align*}
\log\mathbb{P}(t,\alpha,\beta,n)
=&n(n+\beta)\log(1-t)+n\alpha\log t+O_{n,\alpha,\beta}(1-t)+\log G(\alpha+1)\\
&+\log G(n+1)+\log G(n+\beta+1)+\log G(2n+\alpha+\beta+1)\\
&-\log G(n+\alpha+1)-\log G(2n+\beta+1)-\log G(n+\alpha+\beta+1).
\end{align*}
Recall that the Barnes-G function has the following asymptotic expansion \cite{Voros}
\begin{align*}
\log G(z+1)=z^2 \left(\frac{\log z}{2}-\frac{3}{4}\right)+\frac{z}{2}\log (2\pi)-\frac{\log z}{12}+\zeta'(-1)+O(z^{-1}), \quad z\rightarrow\infty,
\end{align*}
where $\zeta(\cdot)$ is the Riemann zeta function.
Applying it to the above expression for $\log \mathbb{P}(t,\alpha,\beta,n)$, we come to the following asymptotic formula.
\begin{theorem} For large $n$ and $t\rightarrow1$, we have
\be\label{Pasy1}
\begin{aligned}
\log\mathbb{P}(t,\alpha,\beta,n)
=&n(n+\beta)\log(1-t)+n\alpha\log t+\log \frac{G(\alpha+1)}{(2\pi)^{\alpha/2}}+\frac{3}{4}\alpha^2\\
&+\left[\frac{n^2}{2}-\frac{1}{12}\right]\log n+\left[\frac{(n+\beta)^2}{2}-\frac{1}{12}\right]\log (n+\beta)\\
&+\left[\frac{(2n+\alpha+\beta)^2}{2}-\frac{1}{12}\right]\log (2n+\alpha+\beta)+\left[\frac{1}{12}-\frac{(n+\alpha)^2}{2}\right]\log(n+\alpha)\\
&+\left[\frac{1}{12}-\frac{(2n+\beta)^2}{2}\right]\log (2n+\beta)+\left[\frac{1}{12}-\frac{(n+\alpha+\beta)^2}{2}\right]\log(n+\alpha+\beta)\\
&+O_{\alpha,\beta}(n^{-1})+O_{n,\alpha,\beta}(1-t),
\end{aligned}
\ee
where $O_{\alpha,\beta}(n^{-1})$ is independent of $t$ and $O_{\alpha,\beta}(n^{-1})\rightarrow0$ as $n\rightarrow\infty$, and for any fixed $n$, $O_{n,\alpha,\beta}(1-t)\rightarrow0$ as $t\rightarrow1$.
\end{theorem}

Now we have the asymptotic expression for $\log\mathbb{P}(t,\alpha,\beta,n)$ with $n\rightarrow\infty$ and $t\rightarrow1$. Noting that
 \[\frac{d}{dt}\log\mathbb{P}(t,\alpha,\beta,n)=\frac{d}{dt}\log D_n(t,\alpha,\beta),\] we will develop in the next section the large $n$ expansion of $\left(d/dt\right)\log\mathbb{P}(t,\alpha,\beta,n)$ by using the equations from \cite{ChenZhang2010} which are related to $\left(d/dt\right)\log D_n(t,\alpha,\beta)$.

\section{Large $n$ Expansion of the Log-derivative of $\mathbb{P}(t,\alpha,\beta,n)$}
Chen and Zhang \cite{ChenZhang2010} studied the Hankel determinant generated by the moments of the Jacobi weight with a jump, i.e.
\[\det\left(\int_0^1x^{i+j}x^{\alpha}(1-x)^{\beta}(A+B\theta(x-t))dx\right)_{i,j=0}^{n-1},\]
with $A\geq0,~ A+B\geq0$, and $\theta(\cdot)$ denoting the Heaviside step function. Note that when $A=0$ and $B=1$, this determinant is reduced to the Hankel determinant $D_n(t,\alpha,\beta)$ of our interest.
In \cite{ChenZhang2010}, the authors made use of the ladder operators adapted to monic orthogonal polynomials and their associated compatibility conditions $(S_1), (S_2)$ and $(S_2')$ to show that the log derivative of the Hankel determinant satisfies a second order differential equation which may be transformed into the $\sigma$-form of a particular Painlev\'{e} VI equation.

Before presenting the results from \cite{ChenZhang2010} which will be used later to study the asymptotic behavior of $D_n(t,\alpha,\beta)$, we first give the definitions of the four auxiliary quantities introduced in \cite{ChenZhang2010}:
\begin{align*}
R_n(t):=&t^{\alpha}(1-t)^{\beta}\frac{P_n^2(t,t)}{h_n(t)},\\
r_n(t):=&t^{\alpha}(1-t)^{\beta}\frac{P_n(t,t)P_{n-1}(t,t)}{h_{n-1}(t)},\\
x_n(t):=&\frac{\beta}{h_n}\int_t^1\frac{P_n^2(x,t)}{1-x}x^{\alpha}(1-x)^{\beta}dx,\\
y_n(t):=&\frac{\beta}{h_{n-1}}\int_t^1\frac{P_n(x,t)P_{n-1}(x,t)}{1-x}x^{\alpha}(1-x)^{\beta}dx.
\end{align*}
Here $P_n(t,t):=P_n(x,t)|_{x=t}$ and $P_n(x,t), n=0,1,2,\cdots,$ are monic orthogonal polynomials with respect to $x^{\alpha}(1-x)^{\beta}$ defined by
\begin{align*}
P_n(x,t)=x^n+p_1(n,t)x^{n-1}+\cdots+P_n(0,t),\\
h_i(t)\delta_{ij}=\int_t^1P_i(x,t)P_j(x,t)x^{\alpha}(1-x)^{\beta}dx.
\end{align*}
In addition, for ease of notations, we write in this section
\[A_n:=2n+1+\alpha+\beta.\]

Below is formula (76) in \cite{ChenZhang2010}, which is crucial for our later derivation of the large $n$ expansion of the log derivative of $D_n(t,\alpha,\beta)$.
\begin{lemma} \label{LWnpvi}
The quantity
\[W_n(t):=1-(1-t)\frac{x_n(t)}{A_n}\]
satisfies the Painlev\'{e} VI equation
\be\label{Wnpvi}
\begin{aligned}
W_n''=&\frac{1}{2}\left(\frac{1}{W_n}+\frac{1}{W_n-1}+\frac{1}{W_n-t}\right)(W_n')^2-\left(\frac{1}{t}+\frac{1}{t-1}+\frac{1}{W_n-t}\right)W_n'\\
&+\frac{W_n(W_n-1)(W_n-t)}{t^2(t-1)^2}\left(\frac{A_n^2}{2}-\frac{\alpha^2t}{2W_n^2}+\frac{\beta^2(t-1)}{2(W_n-1)^2}+\frac{t(t-1)}{2(W_n-t)^2}\right),
\end{aligned}
\ee
with initial conditions
\[W_n(0)=0, \qquad\qquad W_n'(0)=1.\]
\end{lemma}

From \eqref{Wnpvi}, we can derive the asymptotic expansion of $W_n(t)$ in large $n$.

\begin{proposition} For large $n$ and fixed $t$, we have
\be\label{Wnexpan}
\begin{aligned}
W_n(t)=&\frac{\alpha \sqrt{t}}{2n}-\frac{\alpha(1+\alpha+\beta)}{4n^2}\sqrt{t}
+\frac{\alpha\left(t^2(1-4\beta^2)+2t(4(\alpha+\beta)(\alpha+\beta+2)+2\beta^2+3)+1\right)}{64n^3\sqrt{t}}\\
&+O(n^{-4}).
\end{aligned}
\ee
\end{proposition}
\begin{proof}
To begin with, we discuss the reasonable form of the expansion of $W_n(t)$ in large $n$.
Putting $W_n''$ and $W_n'$ to 0 in \eqref{Wnpvi}, we have
\be\label{Wnonder}
W_n(W_n-1)(W_n-t)\left[2n^2+2n(1+\alpha+\beta)+\frac{(1+\alpha+\beta)^2}{2}-\frac{\alpha^2t}{2W_n^2}+\frac{\beta^2(t-1)}{2(W_n-1)^2}+\frac{t(t-1)}{2(W_n-t)^2}\right]=0.
\ee
It gives rise to a sixth order algebraic equation in $W_n$, which we do not know how to solve. However, we can figure out the leading order in the large $n$ expansion of $W_n$ by using \eqref{Wnonder}:
\begin{itemize}
\item[(i)] If $W_n\sim O(n^{-k}),k\geq2$, then $W_n-1\sim O(1)$ and $W_n-t\sim O(1)$, so that $W_n(W_n-1)(W_n-t)\sim O(n^{-k})$ and the term in the square bracket of \eqref{Wnonder} is $O(n^{2k})$. Hence the left hand side of \eqref{Wnonder} is $O(n^{k})$, which obviously contradicts the right hand side of \eqref{Wnonder}.
\item [(ii)] If $W_n\sim O(n^{k}),k\geq1$, then $W_n-1\sim O(n^k)$ and $W_n-t\sim O(n^k)$, so that $W_n(W_n-1)(W_n-t)\sim O(n^{3k})$. The term in the square bracket of \eqref{Wnonder} is $O(n^2)$. Hence the left hand side of \eqref{Wnonder} is $O(n^{3k+2})$, which again leads to a contradiction.
\end{itemize}
As a consequence, we conclude that, for large $n$,
\[W_n\sim O(1) \qquad\text{or}\qquad W_n\sim O(n^{-1}).\]

To continue, we assume
\[W_n(t)=\sum_{i=0}^{\infty}a_i(t)n^{-i},\qquad n\rightarrow\infty.\]
Substituting it into \eqref{Wnpvi}, by taking the large $n$ expansion on both sides and comparing the coefficient of the highest order term in $n$, we find
\[a_0^2(t)(a_0(t)-1)^2(a_0(t)-t)^2=0,\]
which has possible solutions
\[a_0(t)=0,\quad1,\quad t.\]
We note that the solution $a_0(t)=1$ contradicts the initial condition $W_n(0)=0$. If $a_0(t)=t$, then by equating the coefficient in the large $n$ expansion of \eqref{Wnpvi} term by term, we get $a_i(t)=0, i=1,2,\cdots$. Therefore, we choose the solution
\[a_0(t)=0.\]

Taking the series expansion in large $n$ of both sides of \eqref{Wnpvi} again, we find from the leading coefficient
\[4a_1^2(t)-\alpha^2t=0,\]
which has solution
\[a_1(t)=\pm\frac{\alpha}{2}\sqrt{t}.\]
As is given in Lemma \ref{LWnpvi}, we know that $W_n(t)$ satisfies the initial conditions
\[W_n(0)=0,\qquad W_n'(0)=1>0.\]
Hence there exist $\epsilon>0$ such that $W_n(t)>0$ for $t\in(0,\epsilon)$. This indicates that $a_1(t)>0$ for $t\in(0,\epsilon)$. Therefore, we have
\[a_1(t)=\frac{\alpha}{2}\sqrt{t}.\]

Equating the coefficients in the large $n$ expansion of both sides of $\eqref{Wnpvi}$, we can determine $a_i(t),i=2,3,\cdots$ successively.
\end{proof}

Define
\[H_n(t):=t(t-1)\frac{d}{dt}\log D_n(t,\alpha,\beta).\]
Recalling that
\begin{align*}
\mathbb{P}(t,\alpha,\beta,n)=&\frac{D_n(t,\alpha,\beta)}{D_n(0,\alpha,\beta)},
\end{align*}
we readily see that
\[H_n(t)=t(t-1)\frac{d}{dt}\log \mathbb{P}(t,\alpha,\beta,n).\]
To study the large $n$ behavior of $\left(d/dt\right)\log \mathbb{P}(t,\alpha,\beta,n)$, it suffices to express $H_n(t)$ in terms of $W_n(t)$. To achieve this, we first recall formula (56) and (73) in \cite{ChenZhang2010}:
\begin{gather*}
H_n(t)=(2n+\alpha+\beta)(y_n(t)-tr_n(t))-n(n+\alpha),
\end{gather*}
and
\be\label{rnxnyn}
r_n(t)=-\frac{1}{2}+\frac{A_n+(1-t)x_n'(t)}{2x_n(t)}-\frac{\left(2y_n(t)+\beta-(1-t)x_n(t)+t\right)\left(A_n-x_n(t)\right)}{2t\,x_n(t)}.
\ee
A combination of these two identities gives us
\begin{align}\label{Hnxnyn}
H_n(t)=\frac{2n+\alpha+\beta}{2x_n}\left(2A_ny_n-t(1-t)x_n'+(x_n-A_n)\left((1-t)x_n-\beta\right)\right)-n(n+\alpha).
\end{align}
Hence, to get the desired expression for $H_n(t)$ in terms of $x_n(t)$ and $x_n'(t)$, it suffices to represent $y_n(t)$ by $x_n(t)$ and $x_n'(t)$.
\begin{lemma} $y_n(t)$ is expressed in terms of $x_n(t)$ and $x_n'(t)$ by
\begin{align}\label{ynxn}
y_n(t)=&-\frac{x_n\left(x_n\cdot F^2(x_n,x_n',t)+(x_n-A_n)\left((2n+\alpha)\cdot F(x_n,x_n',t)+n(n+\alpha)t\right)\right)}{(2n+\alpha+\beta)(x_n-A_n)((1-t)x_n-A_n)},
\end{align}
where
\[F(x_n,x_n',t):=\frac{1}{2x_n}\left(t(1-t)x_n'+(x_n-A_n)(\beta-(1-t)x_n)\right).\]
\end{lemma}
\begin{proof}
We remind the reader of formula (66) and (67) in \cite{ChenZhang2010}:
\begin{align*}
x_n(t)=&\frac{A_n\left(\widetilde{\ell}(r_n,y_n,t)-t(1-t)y_n'\right)}{2\,k(r_n,y_n,t)},\\
\frac{1}{x_n(t)}=&\frac{\widetilde{\ell}(r_n,y_n,t)+t(1-t)y_n'}{2 A_n(\beta+y_n)y_n},
\end{align*}
where
\begin{align*}
\widetilde{\ell}(r_n,y_n,t):=&2y_n^2+\left(2\beta-(2n+\alpha+\beta+2r_n)t\right)y_n+(2n+\alpha)t r_n+n(n+\alpha)t,\\
k(r_n,y_n,t):=&(y_n-tr_n)^2+\left(\beta-(2n+\alpha+\beta)t\right)y_n+(2n+\alpha)t r_n+n(n+\alpha)t.
\end{align*}
These two equations can be simplified to
\begin{align}
A_n\left(\widetilde{\ell}-t(1-t)y_n'\right)-&2kx_n=0,\label{xnyn1}\\
x_n\left(\widetilde{\ell}+t(1-t)y_n'\right)-&2A_n(\beta+y_n)y_n=0.\label{xnyn2}
\end{align}
To eliminate $y_n'$ from them, we add \eqref{xnyn1} times $x_n$ and \eqref{xnyn2} times $A_n$ and get
\[x_n\left(x_n\cdot k-A_n \cdot\widetilde{\ell}\right)+A_n^2(\beta+y_n)y_n=0.\]
Substituting the definitions of $k$ and $\widetilde{\ell}$ into the above identity, we obtain an equation involving $x_n,r_n$ and $y_n$:
\begin{align*}
x_n(x_n-A_n)\left(tr_n(-2y_n+2n+\alpha)-(2n+\alpha+\beta)ty_n+n(n+\alpha)t\right)&\\
+x_n^2(tr_n)^2+(x_n-A_n)^2y_n(y_n+\beta)&=0.
\end{align*}
Replacing $r_n$ by using \eqref{rnxnyn}, after simplification, we arrive at \eqref{ynxn}.
\end{proof}

Plugging \eqref{ynxn} into \eqref{Hnxnyn}, we find
\begin{align*}
H_n(t)=&-\frac{A_nt^2(1-t)^2(x_n')^2}{4(x_n-A_n)x_n((1-t)x_n-A_n)}+\frac{t(1-t)^2x_n'}{2((1-t)x_n-A_n)}+\frac{1}{4}(1-t)(2n-1+\alpha+\beta)x_n\\
&+\frac{1}{4}\left(tA_n(2n-1+\alpha+\beta)-(2n+\beta)(2n+2\alpha+\beta)-\beta^2\right)+\frac{A_n\beta^2}{4x_n}+\frac{t(\alpha^2-1)A_n}{4\left((1-t)x_n-A_n\right)}.
\end{align*}
Substituting
\[x_n(t)=A_n\frac{1-W_n(t)}{1-t}\]
into this equation, we establish the following expression for $H_n(t)$.

\begin{proposition} $H_n(t)$ is expressed in terms of $W_n(t)$ and $W_n'(t)$ by
\be\label{HnWn}
\begin{aligned}
H_n(t)=&\frac{t^2(1-t)^2(W_n')^2}{4W_n(W_n-1)(W_n-t)}+\frac{t(1-t)W_n'}{2(W_n-t)}-\frac{A_n}{4}(2n-1+\alpha+\beta)W_n\\
&+\frac{1}{4}\left(t((2n+\alpha+\beta)^2+1)+\alpha^2-\beta^2-1\right)-\frac{\alpha^2t}{4W_n}+\frac{\beta^2(t-1)}{4(W_n-1)}+\frac{t(t-1)}{4(W_n-t)}.
\end{aligned}
\ee
\end{proposition}

Inserting the large $n$ expansion of $W_n(t)$, i.e. formula \eqref{Wnexpan}, into \eqref{HnWn}, we come to the expansion of $H_n(t)$.

\begin{theorem} For large $n$ and fixed $t$, we have
\be\label{Hnexpan}
\begin{aligned}
H_n(t)=&n^2t+\left((\alpha+\beta)t-\alpha\sqrt{t}\right)n+\frac{\alpha}{4}\left((\alpha+2\beta)t-2(\alpha+\beta)\sqrt{t}+\alpha\right)\\
&+\frac{\alpha(1-t)(1+(4\beta^2-1)t)}{32\sqrt{t}\;n}+O_{\alpha,\beta}\left(\frac{1}{n^2t}\right),
\end{aligned}
\ee
where $O_{\alpha,\beta}\left(\frac{1}{n^2t}\right)$ depends on $\alpha$ and $\beta$, and for arbitrary $t$, $O_{\alpha,\beta}\left(\frac{1}{n^2t}\right)\rightarrow0$ as $n\rightarrow\infty$.
\end{theorem}

\section{Proof of Theorem \ref{main1} and Corollary \ref{main2}}
\subsection{Proof of Theorem \ref{main1}: Asymptotics of $\mathbb{P}(t,\alpha,\beta,n)$ at the hard edge}
Recall that
\begin{align}\label{HnlogP}
H_n(t)=t(t-1)\frac{d}{dt}\log \mathbb{P}(t,\alpha,\beta,n).
\end{align}
Integrating \eqref{Hnexpan} will provide us with the asymptotic expansion of $\mathbb{P}(t,\alpha,\beta,n)$ in large $n$. Combining it with another asymptotic expression of $\mathbb{P}(t,\alpha,\beta,n)$ obtained in section 2, i.e. formula \eqref{Pasy1}, we derive the asymptotic formula for $\mathbb{P}(t,\alpha,\beta,n)$ with $n$ large and $t$ arbitrary.
\begin{theorem}
For any $t$ and large $n$, we have
\be\label{Pasy2}
\begin{aligned}
\log \mathbb{P}(t,\alpha,\beta,n)=&n^2\log(1-t)+n\left(\beta\log(1-t)+2\alpha\log(1+\sqrt{t})-2\alpha\log2\right)\\
&+\alpha(\alpha+\beta)\log (1+\sqrt{t})-\frac{\alpha^2}{4}\log t-\alpha(\alpha+\beta)\log2+\frac{3}{4}\alpha^2+\log \frac{G(\alpha+1)}{(2\pi)^{\alpha/2}}\\
&+\left[\frac{(n+\beta)^2}{2}-\frac{1}{12}\right]\log (n+\beta)+\left[\frac{n^2}{2}-\frac{1}{12}\right]\log n\\
&+\left[\frac{(2n+\alpha+\beta)^2}{2}-\frac{1}{12}\right]\log (2n+\alpha+\beta)+\left[\frac{1}{12}-\frac{(2n+\beta)^2}{2}\right]\log (2n+\beta)\\
&+\left[\frac{1}{12}-\frac{(n+\alpha)^2}{2}\right]\log(n+\alpha)+\left[\frac{1}{12}-\frac{(n+\alpha+\beta)^2}{2}\right]\log(n+\alpha+\beta)\\
&+O_{\alpha,\beta}(n^{-1})+O_{\alpha,\beta}\left(\frac{1}{n\sqrt{t}}\right),
\end{aligned}
\ee
where $O_{\alpha,\beta}(n^{-1})$ is independent of $t$ and $O_{\alpha,\beta}(n^{-1})\rightarrow0$ as $n\rightarrow\infty$, and for any $t$, $O_{\alpha,\beta}\left(\frac{1}{n\sqrt{t}}\right)\rightarrow0$ as $n\rightarrow\infty$.
\end{theorem}
\begin{proof}
Integrating both sides of \eqref{HnlogP} from $t$ to $t_0$ with $t_0$ near 1, we find
\begin{align*}
\log \mathbb{P}(t_0,\alpha,\beta,n)-\log \mathbb{P}(t,\alpha,\beta,n)=&\int_t^{t_0}\frac{d}{d\xi}\log \mathbb{P}(\xi,\alpha,\beta,n)d\xi=\int_t^{t_0}\frac{H_n(\xi)}{\xi(\xi-1)}d\xi,
\end{align*}
namely,
\begin{align*}
\log \mathbb{P}(t,\alpha,\beta,n)-\log \mathbb{P}(t_0,\alpha,\beta,n)
=\int_t^{t_0}\frac{H_n(\xi)}{\xi(1-\xi)}d\xi.
\end{align*}
Hence, according to \eqref{Hnexpan}, we obtain
\be\label{Ptt0}
\begin{aligned}
\log \mathbb{P}(t,\alpha,\beta,n)&-\log \mathbb{P}(t_0,\alpha,\beta,n)\\
=&n(n+\beta)\log\frac{1-t}{1-t_0}+\alpha(2n+\alpha+\beta)\log\frac{1+\sqrt{t}}{1+\sqrt{t_0}}+\frac{\alpha^2}{4}\log\frac{t_0}{t}\\
&
+O_{\alpha,\beta}\left(\frac{1}{n\sqrt{t}}\right)+O_{\alpha,\beta}\left(\frac{1}{n\sqrt{t_0}}\right).
\end{aligned}
\ee
Replacing $t$ by $t_0$ in \eqref{Pasy1}, we have
\begin{align*}
\log\mathbb{P}(t_0,\alpha,\beta,n)
=&n(n+\beta)\log(1-t_0)+n\alpha\log t_0+\log \frac{G(\alpha+1)}{(2\pi)^{\alpha/2}}+\frac{3}{4}\alpha^2\\
&+\left[\frac{n^2}{2}-\frac{1}{12}\right]\log n+\left[\frac{(n+\beta)^2}{2}-\frac{1}{12}\right]\log (n+\beta)\\
&+\left[\frac{(2n+\alpha+\beta)^2}{2}-\frac{1}{12}\right]\log (2n+\alpha+\beta)+\left[\frac{1}{12}-\frac{(n+\alpha)^2}{2}\right]\log(n+\alpha)\\
&+\left[\frac{1}{12}-\frac{(2n+\beta)^2}{2}\right]\log (2n+\beta)+\left[\frac{1}{12}-\frac{(n+\alpha+\beta)^2}{2}\right]\log(n+\alpha+\beta)\\
&+O_{\alpha,\beta}(n^{-1})+O_{n,\alpha,\beta}(1-t_0).
\end{align*}
Substituting it into \eqref{Ptt0}, we see that the term $n(n+\beta)\log(1-t_0)$ disappears. By sending $t_0\rightarrow1$ in the resulting expression, we complete our proof.
\end{proof}

Putting
\[t=\frac{s}{4n^2}\]
in \eqref{Pasy2} and taking the series expansion of its right hand side in large $n$, by sending $n$ to $\infty$, we find
\begin{align}\label{limPs}
\lim_{n\rightarrow\infty}\log \mathbb{P}\left(\frac{s}{4n^2},\alpha,\beta,n\right)=&\log\det(I-K_{\mathrm{Bessel}})\nonumber\\
=&-\frac{s}{4}+\alpha\sqrt{s}-\frac{\alpha^2}{4}\log s+\log \frac{G(\alpha+1)}{(2\pi)^{\alpha/2}}+O(s^{-1/2}),\qquad s\rightarrow\infty.
\end{align}
The subsequent terms in \eqref{limPs} was derived in the joint work of the authors and Fan \cite[formula (5.12)]{LyuChenFan}. However, the constant term was not identified there. Combining formula (5.12) therein with our formula \eqref{limPs}, we come to the full asymptotic expression \eqref{Pfullexpan}. This completes the proof of Theorem \ref{main1}.

\begin{remark}
It turns out that, at the hard edge, the constant term in the asymptotics of the smallest eigenvalue distribution of JUE with the weight $x^{\alpha}(1-x)^{\beta}$ is independent of $\beta$.
Alternatively, we may see the $\beta-$independence by using the following definition of $\mathbb{P}(t,\alpha,\beta,n)$:
\[\mathbb{P}(t,\alpha,\beta,n)=\frac{\frac{1}{n!}\int_{[t,1]^n}\prod\limits_{1\leq i<j\leq n}(x_i-x_j)^2\prod\limits_{k=1}^n x_k^{\alpha}(1-x_k)^{\beta}dx_k}{\frac{1}{n!}\int_{[0,1]^n}\prod\limits_{1\leq i<j\leq n}(x_i-x_j)^2\prod\limits_{k=1}^n x_k^{\alpha}(1-x_k)^{\beta}dx_k}.\]
By setting $\alpha=0$ and via a change of variables $x_{\ell}=t+(1-t)y_{\ell},\ell=1,2,\cdots,n$, we get
\[\mathbb{P}(t,0,\beta,n)=(1-t)^{n(n+\beta)},\]
so that
\[\lim_{n\rightarrow\infty}\log\mathbb{P}\left(\frac{s}{4n^2},0,\beta,n\right)=-\frac{s}{4},\]
which is independent of $\beta$. This may indicate the $\beta-$independence of the constant term in the asymptotic expansion of $\log\mathbb{P}\left(\frac{s}{4n^2},\alpha,\beta,n\right)$ in large $s$ with $n\rightarrow\infty$.
\end{remark}

\begin{remark}
The largest eigenvalue distribution $\mathbb{P}_L(t,\alpha,\beta,n)$, i.e. the probability that all eigenvalues of JUE with the weight $x^{\alpha}(1-x)^{\beta}$ is not greater than $t$, is connected with the smallest eigenvalue distribution of our interest, i.e. $\mathbb{P}(t,\alpha,\beta,n)$, by
\[\mathbb{P}_L(t,\alpha,\beta,n)=\mathbb{P}(1-t,\beta,\alpha,n).\]
Here we remind the reader to note the exchange of $\alpha$ and $\beta$. Then, at the hard edge, we have
\[\mathbb{P}_L\left(1-\frac{s}{4n^2},\alpha,\beta,n\right)=\mathbb{P}\left(\frac{s}{4n^2},\beta,\alpha,n\right).\]
Therefore, as $n\rightarrow\infty$, it follows from \eqref{Pfullexpan} that the probability that all eigenvalues of JUE lie in the interval $[0,1-s/(4n^2)]$ has the following asymptotic expression
\begin{align*}
\lim_{n\rightarrow\infty}\log \mathbb{P}_L&\left(1-\frac{s}{4n^2},\alpha,\beta,n\right)=\log\det(I-K_{\mathrm{Bessel}}^{(\beta)})\\
=&-\frac{s}{4}+\beta \sqrt{s}-\frac{\beta^2}{4}\log s+\log \frac{G(\beta+1)}{(2\pi)^{\beta/2}}+\frac{\beta }{8}s^{-1/2}\\
&+\frac{\beta
   ^2}{16}s^{-1}
   +\left(\frac{\beta^3}{24}+\frac{3\beta}{128}\right)s^{-3/2}+\left(\frac{\beta^4}{32}+\frac{9\beta^2}{128}\right)s^{-2}\\
&
+\left(\frac{\beta^5}{40}+\frac{9\beta^3}{64}+\frac{45\beta}{1024}\right)s^{-5/2}+O\left(s^{-3}\right),\qquad s\rightarrow\infty.
\end{align*}
Here $K_{\mathrm{Bessel}}^{(\beta)}$ is the integral operator with the kernel
\[K_{\mathrm{Bessel}}^{(\beta)}(x,y):=\frac{J_{\beta}(\sqrt{x})\sqrt{y}J_{\beta}'(\sqrt{y})-\sqrt{x}J_{\beta}'(\sqrt{x})J_{\beta}(\sqrt{y})}{2(x-y)},\]
which is the Bessel kernel $K_{\mathrm{Bessel}}(x,y)$ defined by \eqref{Besal} with $\alpha$ replaced by $\beta$.
We observe that the constant term in the above asymptotics now reads
\[\log\frac{G(\beta+1)}{(2\pi)^{\beta/2}},\]
 which is independent of $\alpha$.
\end{remark}

\begin{remark}
As we know, the hard edge scaling limit of the smallest eigenvalue distribution of Laguerre unitary ensembles (LUE for short) is also the Fredholm determinant of the Bessel kernel \cite{Forrester1993, TW161}. Denoting by $\widetilde{\mathbb{P}}(t,\alpha,n)$ the probability that all eigenvalues of LUE are in the interval $[t,\infty)$, we have
\[\lim_{n\rightarrow\infty}\widetilde{\mathbb{P}}\left(\frac{s}{4n},\alpha,n\right)=\det\left(I-K_{\mathrm{Bessel}}\right),\]
Based on the finite $n$ results obtained via the ladder operator approach and by using Dyson's Coulomb fluid method, the authors and Fan \cite{LyuChenFan} showed that, for sufficiently large $n$, the constant term in the asymptotic expansion of $\widetilde{\mathbb{P}}\left(\frac{s}{4n},\alpha,n\right)$ in large $s$ satisfies a first order difference equation in $\alpha$, i.e.
\[c_1(\alpha+1)-c_1(\alpha)=\log\left(\frac{\Gamma(\alpha+1)}{\sqrt{2\pi}}\right),\]
which has solutions
\begin{align}\label{c1c0}
c_1(\alpha)=\log\left(\frac{G(\alpha+1)}{(2\pi)^{\alpha/2}}\right)+c_0,
\end{align}
where $c_0$ is a constant independent of $\alpha$. To determine the value of $c_0$, we need an initial condition for $c_1(\alpha)$, which regrettably was not provided in \cite{LyuChenFan}.

We observe that the initial value $c_1(0)$ can be obtained via arguments similar to the ones used in the previous remark for the JUE case. Actually, by definition, we know that $\widetilde{\mathbb{P}}(t,\alpha,n)$ is the quotient of two multiple integrals, namely,
\[\widetilde{\mathbb{P}}(t,\alpha,n)=\frac{\frac{1}{n!}\int_{[t,\infty)^n}\prod\limits_{1\leq i<j\leq n}(x_i-x_j)^2\prod\limits_{k=1}^n x_k^{\alpha}\e^{-x_k}dx_k}{\frac{1}{n!}\int_{[0,\infty)^n}\prod\limits_{1\leq i<j\leq n}(x_i-x_j)^2\prod\limits_{k=1}^n x_k^{\alpha}\e^{-x_k}dx_k}.\]
Setting $\alpha=0$ and by changing variables $x_{\ell}=y_{\ell}+t, \ell=1,2,\cdots,n$, we find
\[\widetilde{\mathbb{P}}(t,0,n)=\e^{-nt},\]
so that
\[\lim_{n\rightarrow\infty}\widetilde{\mathbb{P}}\left(\frac{s}{4n},0,n\right)=\e^{-s/4}.\]
This combined with formula (3.14) of \cite{LyuChenFan} leads to
\[c_1(0)=0=c_0.\]
Hence, it follows from \eqref{c1c0} that
\[c_1(\alpha)=\log\left(\frac{G(\alpha+1)}{(2\pi)^{\alpha/2}}\right),\]
which also gives the constant term in the asymptotics of the Fredholm determinant of the Bessel kernel.
\end{remark}

\subsection{Proof of Corollary \ref{main2}}
We denote by $\widetilde{\mathbb{P}}(a,\beta,n)$ the probability that the interval $(-a,a),a>0,$ contains no eigenvalues of JUE with the weight $(1-x^2)^{\beta}, x\in[-1,1],\beta>0$. As is illustrated in \cite[proof of Theorem 4.1]{LyuChenFan}, under the scaling
\[b=na,\]
 and as $n\rightarrow\infty$, $\widetilde{\mathbb{P}}(a,\beta,n)$ is connected with $\mathbb{P}(t,\alpha,\beta,n)$ by the relation
\begin{align}\label{JUEsym}
\lim_{n\rightarrow\infty}\widetilde{\mathbb{P}}\left(\frac{b}{n},\beta,n\right)=\lim_{n\rightarrow\infty}\mathbb{P}\left(\frac{b^2}{4n^2},-\frac{1}{2},\beta,n\right)+\lim_{n\rightarrow\infty}\mathbb{P}\left(\frac{b^2}{4n^2},\frac{1}{2},\beta,n\right).
\end{align}

By setting $s=b^2$ and $\alpha=\pm1/2$ in \eqref{Pfullexpan}, we find as $b\rightarrow\infty$
\begin{equation*}
\begin{aligned}
\lim_{n\rightarrow\infty}\log \mathbb{P}\left(\frac{b^2}{4n^2},-\frac{1}{2},\beta,n\right)
=&-\frac{b^2}{4}-\frac{b}{2}-\frac{1}{8}\log b+\log \left( G\left(1/2\right)\cdot(2\pi)^{1/4}\right)-\frac{1}{16b}\\
&+\frac{1}{64b^2}-\frac{13}{768b^3}+\frac{5}{256b^4}-\frac{413}{10240b^5}+O\left(b^{-6}\right),
\end{aligned}
\end{equation*}
and
\begin{equation*}
\begin{aligned}
\lim_{n\rightarrow\infty}\log \mathbb{P}\left(\frac{b^2}{4n^2},\frac{1}{2},\beta,n\right)
=&-\frac{b^2}{4}+\frac{b}{2}-\frac{1}{8}\log b+\log \left( \frac{G\left(3/2\right)}{(2\pi)^{1/4}}\right)+\frac{1}{16b}\\
&+\frac{1}{64b^2}+\frac{13}{768b^3}+\frac{5}{256b^4}+\frac{413}{10240b^5}+O\left(b^{-6}\right).
\end{aligned}
\end{equation*}
It follows from \eqref{JUEsym} that
\begin{equation}\label{PJs1}
\begin{aligned}
\lim_{n\rightarrow\infty}\log\widetilde{\mathbb{P}}\left(\frac{b}{n},\beta,n\right)
=&-\frac{b^2}{2}-\frac{\log b }{4}+\log\left(G(1/2)\cdot G(3/2)\right)+\frac{1}{32b^2}\\
&+\frac{5}{128b^4}+O(b^{-6}),\qquad b\rightarrow\infty.
\end{aligned}
\end{equation}
Noting that
\[ G\left(3/2\right)=G\left(1/2\right)\Gamma\left(1/2\right),\qquad\Gamma\left(1/2\right)=\sqrt{\pi},\]
and $G\left(1/2\right)$ is given by \cite{Voros}
\[ G\left(1/2\right)={\rm e}^{3\zeta'(-1)/2}\pi^{-1/4}2^{1/24},\]
we have
\[G(1/2)\cdot G(3/2)={\rm e}^{3\zeta'(-1)}2^{1/12}.\]
Combining it with \eqref{PJs1}, we get \eqref{PJs}.

\begin{remark}
We observe that the first three terms in the expansion \eqref{PJs} are the same as the ones in the asymptotics of the Fredholm determinant of the sine kernel \cite{DIKZ2007, Ehrhardtsine, Krasovsky2004}.
\end{remark}

\section*{Acknowledgments}
Shulin Lyu was supported by National Natural Science Foundation of China under grant number 11971492. Yang Chen was supported by the Macau Science and Technology Development Fund under grant number FDCT 023/2017/A1 and by the University of Macau under grant number MYRG 2018-00125-FST.

\end{document}